\documentclass[11pt]{article}

\usepackage[margin=1in]{geometry}
\usepackage{amsmath,amssymb,amsthm,mathtools}
\usepackage[T1]{fontenc}
\usepackage{lmodern,microtype}
\usepackage{xurl}
\usepackage[sort]{cite}
\usepackage[colorlinks=true,linkcolor=blue,urlcolor=blue,citecolor=blue]{hyperref}

\newtheorem{theorem}{Theorem}
\newtheorem{lemma}[theorem]{Lemma}

\newcommand{\Ot}{\widetilde O}
\newcommand{\E}{\mathbb E}
\newcommand{\Var}{\operatorname{Var}}

\DeclareMathOperator{\polylog}{polylog}

\title{Maximum Matching Size for Bounded Arboricity Graphs \\
in the Dynamic Graph Stream Model using $\tilde{O}(n^{2/3})$ space}
\author{Andrew McGregor\thanks{This author is supported by the National Science Foundation under grant CCF-2521579.}}
\date{}

\begin{document}
\maketitle

\vspace{-.2in}
\begin{abstract}
The paper presents a one-pass algorithm in the insert-delete graph stream model that  returns a  $(1+\varepsilon)(\alpha+2)$-approximation for the size of the  maximum
matching in a graph of arboricity at most $\alpha$. The algorithm uses
$O(\varepsilon^{-4/3}\alpha^{4/3}n^{2/3} \polylog n)$ space.
For constant $\alpha$ and $\varepsilon$, this improves the best known
previous space bound from $O(n^{4/5}  \polylog n)$ to $O(n^{2/3}  \polylog n)$.
The algorithm is a linear sketch and requires no bounds on the number of deletions or on the
arboricity of intermediate graphs. 
\end{abstract}

\section{Introduction}
Graph matchings have been studied in the data stream model for over
twenty years~\cite{FeigenbaumKMSZ04}, including work in the (adversarial order)
insertion-only
model~\cite{McGregor05,EpsteinLMS11,Kliemann11,EggertKMS12,Zelke12,
GoelKK12,KonradMM12,Kapralov13,AhnG13,CrouchS14,
Assadi:Khanna:Li:17,KaleT17,Konrad18,Tirodkar18,
PazS19,GhaffariW19,GamlathKMS19,AssadiKSY20,
Kapralov21,AssadiN21,AssadiLT21,KonradN21,
AssadiJJST22,FischerMU22,FeldmanS22,
AssadiBKL23,AssadiS23,KonradN24,FerdousPH24,
Assadi25,Mitrovic0SS25,AssadiJX26,abs-2607-14656},
the random-order 
model~\cite{KonradMM12,KapralovKS14,MonemizadehMPS17,Konrad18,
AssadiBBMS19,GamlathKMS19,0001HMRR20,Bernstein20,
AssadiB21,AssadiS23Random,HashemiW24},
and the dynamic model, where edges can be both inserted and
deleted~\cite{ChitnisCHM15,Konrad15,
Assadi:Khanna:Li:Yaroslavtsev:16,ChitnisCEHMMV16,
Assadi:Khanna:Li:17,DarkK20,ChenHK0X21,
AssadiS22,AssadiB0NS26,Khanna0D25}.
Two longstanding questions about graph matchings have recently been resolved under the \emph{semi-streaming} memory restriction, where algorithms use $O(n \cdot \polylog n)$ space and $n$ is the number of nodes. In dynamic streams, $\Theta(\log\log n)$ passes are necessary and sufficient to construct a matching within a constant factor of optimal~\cite{AssadiB0NS26}. In adversarial-order insertion-only streams, no single-pass algorithm can achieve a $(1/2+\varepsilon)$-approximation for any fixed $\varepsilon>0$, even with randomization, establishing the optimality of the greedy algorithm~\cite{abs-2607-14656}.

There has also been considerable work focused on graphs with bounded arboricity \cite{EsfandiariHLMO18,
      CormodeJMM17,0001V18,BuryGMMSVZ19,ChenCEW23,00010ST24,
      Assadi:Khanna:Li:17,Jowhari23,GhorbaniJ24}. Recall that a graph has arboricity at most $\alpha$ if every subgraph on $t\ge1$ vertices has at most $\alpha(t-1)$ edges. Arboricity is a more natural measure of sparsity than the ratio of edges to nodes. For constant arboricity graphs, it is possible to store the entire graph in the semi-streaming setting and yet the matching could still be as large as $n/2$. Hence, rather than finding a large matching in a bounded arboricity graph, the focus in the literature is on estimating the size of the maximum matching $\mu$ in $o(n)$ space.
      
\paragraph{Our Result.} We consider the problem in the dynamic graph stream model, i.e., the stream consists of an arbitrary number of edge insertions and deletions and is fixed independently of the algorithm's random bits. Only the final graph $G$ is required to have arboricity at most $\alpha$. The result of the paper is:

\begin{theorem}\label{thm:main}
For $\varepsilon \in (0,1)$, there is a one-pass randomized algorithm in the insert-delete graph stream model
using $O(\varepsilon^{-4/3}\alpha^{4/3}n^{2/3} \polylog n)$ space that returns $M_{\rm est}$ such that,
with high probability\footnote{Here, and throughout the paper, \emph{with high probability} means probability at least $1-n^{-c}$ for any prescribed constant $c>0$, with constants in the space bound depending on $c$.}, 
\[
 \mu\le M_{\rm est}\le(1+\varepsilon)(\alpha+2)\mu.
\]
\end{theorem}

\paragraph{Previous Work.}
The best previous one-pass space bound for an $O(\alpha)$ approximation in the insert-delete model
was
$\tilde{O}_{\alpha,\varepsilon}(n^{4/5})$~\cite{BuryGMMSVZ19,ChitnisCEHMMV16}. 
In contrast, Assadi et al.~\cite{Assadi:Khanna:Li:17} proved an
$\Omega(\sqrt{n}/r^{5/2})$-bit lower bound for randomized
one-pass algorithms that estimate maximum matching size
within a factor $r$ in dynamic streams, even on graphs
of arboricity $O(r)$.

Various results were known in more relaxed models. Cormode et al.~\cite{CormodeJMM17} gave a one-pass
$O(\alpha)$-approximation using
$\tilde{O}_{\alpha,\varepsilon}(n^{2/3})$ space in the insert-delete model under the
additional assumption that the number of deletions is
$O(\alpha n)$.
In the unrestricted dynamic model, Konrad et al.~\cite{00010ST24}
gave $(\alpha+2)(1+\varepsilon)$-approximation algorithms
using two passes and $\tilde{O}_{\alpha,\varepsilon}(n^{3/5})$
space, or three passes and
$\tilde{O}_{\alpha,\varepsilon}(n^{1/2})$ space.
In the insertion-only model, Cormode et al.~\cite{CormodeJMM17}
gave the first $O(\alpha)$-approximation using polylogarithmic
space for constant $\alpha$.
McGregor and Vorotnikova~\cite{0001V18} improved the approximation
factor to $(\alpha+2)(1+\varepsilon)$ and the space bound to
$O(\varepsilon^{-2}\log^2 n)$ bits. In the adjacency-list model, where edges arrive grouped by endpoint (every edge appears twice), Bury  et al.~\cite{BuryGMMSVZ19} gave a deterministic $(\alpha+2)$-approximation using $O(\log n)$ space and Ghorbani and Jowhari~\cite{GhorbaniJ24,Jowhari23} showed that the approximation factor could be reduced to $(\alpha+1)(1+\varepsilon)$, for integer $\alpha\geq 2$, at the expense of increasing the space to $\tilde{O}_{\alpha,\varepsilon}(\sqrt{n})$.

\paragraph{Notation and Assumptions.} Write $m$ for the final edge count, $d_v$ for the final degree of vertex $v$, and $\mu=\mu(G)$ for the maximum matching size. We consider simple undirected graphs on $V=[n]$.
The algorithm is given $n$ and an integer upper bound
$\alpha\ge1$ on the arboricity of the final graph. We use $\Ot$ to suppress factors polylogarithmic in $n$. We assume $m>0$ throughout. Note that the value of $m$ can be computed exactly in parallel and if $m=0$ then $\mu=0$.

\section{Technical Overview}\label{sec:framework}

\subsection{A Useful Degree Statistic Revisited}
Our starting point is the following degree statistic analyzed by Bury et al.\footnote{They write the  statistic as 
$\sum_v\min\{d_v/2,\alpha+1-d_v/2\}
=m-\sum_v\max(d_v-\alpha-1,0)$ but this equals $Q$.} \cite{BuryGMMSVZ19}.
\begin{equation}\label{eq:twoforms}
 Q=m-\sum_v\max(d_v-D,0)
   =\sum_v\min\{d_v,D\}-m \quad \mbox{ where } \quad D=\alpha+1 \ .
\end{equation}
The following notation will be convenient throughout the paper:
\[
 f_v=\max\{0,d_v-D\},\qquad F=\sum_v f_v
\]
and so $Q$ can be written as $Q=m-F$.

\begin{lemma}[Bury et al.~\cite{BuryGMMSVZ19}, Theorem 9 and its proof]\label{lem:statistic}
Let $G$ be a simple graph of arboricity at most $\alpha$. Then,
\begin{equation}\label{eq:positive-count}
 \mu(G)\le Q\le(\alpha+2)\mu(G)
\end{equation}
and $Q\ge h_0:=|\{v:f_v>0\}|$.
\end{lemma}

\subsection{Estimation via Recursive Sketching}
The final algorithm will be based on computing three different estimates. The accuracy of these will depend on the size of $\mu$ relative to the following parameter:
\begin{equation}\label{eq:parameters}
 K=\left\lceil
     \left(C_0\eta^{-2}\alpha^2n\log_2^2(2n)\right)^{1/3}
   \right\rceil ,
\end{equation}
where $C_0>0$ is a sufficiently large constant and $\eta=\varepsilon/10$.
The main contribution of this paper is a new way to build two estimates that will be relevant when $\mu$ is relatively large. The following lemma (proved in Section~\ref{sec:sampling}) states the guarantees of these estimates.

\begin{lemma}[Recursive and Terminal Estimators]\label{lem:degree-estimators}
For $K$ as in~\eqref{eq:parameters}, there is a one-pass algorithm
using $\Ot(K^2)$ space that returns
$\widehat Q_{\rm rec}$ and $\widehat Q_{\rm term}$ satisfying, with
high probability,
\begin{align}
 |\widehat Q_{\rm rec}-Q|
 &\le\eta\max\{K,Q\}
 &&\text{if }\mu\le K^2,\label{eq:recaccuracy}\\
 |\widehat Q_{\rm term}-Q|
 &\le\eta K^2
 &&\text{for every }\mu.\label{eq:termaccuracy}
\end{align}
\end{lemma}

For the sake of intuition, note that 
\begin{eqnarray*}
K \leq \mu\leq K^2 & \Rightarrow &  \widehat Q_{\rm rec}=(1\pm \eta)Q \\
\mu\geq K^2 & \Rightarrow &  \widehat Q_{\rm term}=(1\pm \eta)Q 
\end{eqnarray*}
In the next subsection, we will discuss estimation when $\mu\leq K$, along with how to use these estimates without knowing which range we are in.

\paragraph{Recursive Sketching.} Our approach is based on the 
recursive-sketching framework of Braverman and Ostrovsky~\cite{BravermanO13a}. The basic idea is to consider a sequence of random subsets of nodes  $V=V_0\supseteq V_1\supseteq V_2\supseteq \ldots \supseteq V_{L}$ where each set is roughly half the size of the previous set and $V_{L}$ has $O(K^2)$ nodes in expectation. For each level, define $F_i=\sum_{v\in V_i} f_v$. The algorithm computes $F_L$ exactly and uses the same stored
degrees to form the ``terminal estimate''; the median of independent
copies yields $\widehat Q_{\rm term}$.  For $i<L$,  $F_i$
is estimated based on an estimate of $F_{i+1}$ and an additional ``correction'' term based on the high-degree vertices in $V_i$; these high-degree terms will be identified via CountSketch~\cite{CharikarCF04}.  The goal of these correction terms is to reduce the
error accumulated in estimating $F=F_0$. Subtracting this  estimate from $m$ is the basis for the $\widehat Q_{\rm rec}$ estimate.

\subsection{Combining Estimates (Proof of Theorem~\ref{thm:main})}\label{sec:ranges}

In parallel with the algorithm referenced in Lemma \ref{lem:degree-estimators}, we run an algorithm by Chitnis et al.~\cite{ChitnisCEHMMV16} that has the following guarantees. (See also  Bury et al.~\cite{BuryGMMSVZ19} for a similar algorithm.)

\begin{lemma}[Chitnis et al.~\cite{ChitnisCEHMMV16}]\label{lem:small-matching}
For $K$ as in~\eqref{eq:parameters}, there is a one-pass algorithm
using $\Ot(K^2)$ space that outputs a matching of $G$. Writing
$\widehat Q_{\rm small}$ for its size, with high probability
\begin{equation}\label{eq:smallaccuracy}
 \widehat Q_{\rm small}=\mu \qquad\text{if }\mu\le K.
\end{equation}
\end{lemma}

Given the three estimators, the final output of the algorithm is:
\begin{equation}\label{eq:output}
 M_{\rm est}=\frac{\max\left\{
       \min\{K,\widehat Q_{\rm small}\},\;
       \min\{K^2,\widehat Q_{\rm rec}-\eta K\},\;
       \widehat Q_{\rm term}-\eta K^2
     \right\}}{1-2\eta} \ .
\end{equation}
The role of capping the estimates and/or subtracting terms is to take into account that the accuracy guarantees of some of the estimates do not hold for all values of $\mu$.
 
\begin{proof}[Proof of Theorem~\ref{thm:main}]
Assume that
\eqref{eq:termaccuracy} holds, 
\eqref{eq:smallaccuracy} holds when $\mu \leq K$, and that
\eqref{eq:recaccuracy}  holds when $\mu\le K^2$.
By a union bound, these applicable guarantees hold simultaneously
with high probability. 

For the upper bound, we bound each term as follows:
\begin{align*}
 \min\{K,\widehat Q_{\rm small}\}
  \le \min\{K,\mu\}\le\mu & \le Q \\
   \widehat Q_{\rm term}-\eta K^2 &\le  Q\\
  \min \{ K^2, \widehat Q_{\rm rec}-\eta K\} &\le (1+\eta) Q
\end{align*}
where the last line follows because if $\mu>K^2$, then $K^2<\mu\le Q$ and if
$\mu\le K^2$, then~\eqref{eq:recaccuracy} gives
\[
 \widehat Q_{\rm rec}-\eta K
 \le Q+\eta\max\{K,Q\}-\eta K
 \le(1+\eta)Q.
\]
For the lower bound, consider the three ranges of $\mu$:
  \begin{eqnarray*}
\mu \leq K &  \Rightarrow & \min\{K, \widehat Q_{\rm small}\} = \mu   \\
K\le\mu\le K^2
&\Rightarrow&
\min\{K^2,\widehat Q_{\rm rec}-\eta K\}
\ge \min\{\mu,(1-\eta)Q-\eta K\}
\ge (1-2\eta)\mu \\
\mu\geq K^2 & \Rightarrow &  \widehat Q_{\rm term}-\eta K^2 \geq (1- \eta)Q-\eta \mu \geq  (1- 2\eta)\mu  
\end{eqnarray*}

Using Lemma~\ref{lem:statistic}, we obtain
\[
 \mu\le M_{\rm est}
 \le\frac{1+\eta}{1-2\eta}Q
 \le\frac{1+\eta}{1-2\eta}(\alpha+2)\mu 
 \le (1+\varepsilon) (\alpha+2) \mu.
\]

All three estimators use one pass and $\Ot(K^2)=
\Ot\bigl(\varepsilon^{-4/3}\alpha^{4/3}n^{2/3}\bigr)$ space.
\end{proof}

\section{The Terminal and Recursive Estimators}\label{sec:sampling}

In this section, we design the terminal and recursive estimators, $\widehat Q_{\rm term}$ and  $\widehat Q_{\rm rec}$, and prove Lemma \ref{lem:degree-estimators}.
Both estimators are based on node sub-sampling. We define a sequence of subsets of nodes $V=V_0\supseteq V_1 \supseteq \ldots \supseteq V_L$ using pairwise independent hash functions, where
\begin{equation}\label{eq:depth}
 L=\max\{0,\lceil\log_2(n/K^2)\rceil\}.
\end{equation}
Specifically, let $b_1,\ldots,b_L:V\rightarrow\{0,1\}$ be pairwise-independent
hash functions, chosen independently of one another. Define
\[
 V_0=V,\qquad
 V_{i+1}=\{v\in V_i:b_{i+1}(v)=1\},\qquad
 p_i=2^{-i},\qquad
 F_i=\sum_{v\in V_i}f_v.
\]
In particular, $F_0=F$. 
Our choice of $L$
gives $\E[|V_L|]=np_L\le K^2$. We will compute the exact degrees of all nodes in $V_L$.
 In Section~\ref{sec:terminal}, we discuss how to use the degrees of nodes in $V_L$ to construct the estimator $\widehat Q_{\rm term}$.
In Sections~\ref{sec:discovery} and~\ref{sec:recursive}, we construct $\widehat Q_{\rm rec}$ by  estimating $F$ based on a recursive procedure starting with $F_L$ (which can be computed exactly given the degrees of nodes in $V_L$) and using additional information about the high-degree nodes in each $V_i$.

\subsection{The Terminal Estimator} \label{sec:terminal}

Recall from~\eqref{eq:twoforms} that $Q$ can be written as   $Q=\sum_v \min \{d_v,D\} - m$. This can be estimated via

\begin{equation}\label{eq:terminalestimate}
 \widehat q_{\rm term}
 =p_L^{-1}\sum_{v\in V_L}\min(d_v,D)-m 
\end{equation}
and note that
\begin{equation}\label{eq:terminalvariance}
 \E(\widehat q_{\rm term})=Q,\qquad
 \Var(\widehat q_{\rm term})
 =\frac{1-p_L}{p_L}\sum_v\min(d_v,D)^2
 \le\frac{D^2n(1-p_L)}{p_L}
 \le\frac{8\alpha^2n^2}{K^2}.
\end{equation}
For $L\ge1$, the last inequality uses $p_L>K^2/(2n)$ and $D\le2\alpha$.
For $L=0$, the variance is zero.

\paragraph{Reducing Error Probability.} By an application of the Markov bound and \eqref{eq:terminalvariance}, for every
$\mu$,
\[
 \Pr\left[
   |\widehat q_{\rm term}-Q|>\eta K^2
 \right]
 \le
 \frac{\mathbb E[(\widehat q_{\rm term}-Q)^2]}
      {\eta^2K^4}
 \le
 \frac{8\alpha^2n^2}{\eta^2K^6}
 \le\frac1{10}
\]
by setting the constant $C_0$ in the definition of $K$ sufficiently high. Hence, letting $\widehat Q_{\rm term}$ be  the median of $O(\log n)$ independent copies of $\widehat q_{\rm term}$ ensures that $|\widehat Q_{\rm term}-Q| \leq \eta K^2$ with high probability. The space to compute each $\widehat q_{\rm term}$ is the $\tilde{O}(L)$ space to store  the $b_i$ hash functions and  $\tilde O (|V_L|)$ space required to compute degrees in $V_L$. Recall $\E[|V_L|] \leq K^2$. To ensure a worst-case space bound, we note that if we return an arbitrary value for $\widehat q_{\rm term}$ whenever $|V_L|>10K^2$ then the probability  $|\widehat q_{\rm term}-Q|>\eta K^2$ is at most $1/10+1/10=1/5$ and we still have $|\widehat Q_{\rm term}-Q| \leq \eta K^2$ with high probability. Hence, we can avoid computing the degrees of nodes in $V_L$ when there are many such nodes and the space is $\tilde{O}(K^2)$  in the worst case.

\subsection{Degree Estimation via CountSketch}\label{sec:discovery}
For the recursive estimator, we will use CountSketch to estimate the degrees of nodes in each $V_i$.
For $0\le i<L$, define
\begin{equation}\label{eq:threshold}
 \tau_i:=c_\tau\alpha\max\{1,\sqrt{|V_i|/K^2}\} \qquad \mbox{ where } c_\tau:=32 \ .
\end{equation}

The CountSketch algorithm \cite{CharikarCF04} is designed to estimate the frequency of elements in a data stream. Specifically, we consider a stream whose $t$-th element is $a_t=(i_t,\Delta_t)\in V\times\{-1,1\}$ and for each $v\in V$, define the frequency of $v$ to be $x_v=\sum_{t:i_t=v}\Delta_t$. The CountSketch algorithm uses $O\!\left(w\log(n/\delta)\right)$ counters and returns estimates $\tilde{x}_1,\ldots,\tilde{x}_n$ such that with probability at least $1-\delta$,
\begin{equation}\label{eq:cm}
 |\tilde{x}_v-x_v|\leq\frac{\|\mathbf{x}_{w\textup{-tail}}\|_2}{\sqrt{w}} \qquad \mbox{for all $v\in V$}
\end{equation}
where $\|{\bf x}_{w\textup{-tail}}\|_2$ is the $\ell_2$ norm of the vector of frequencies with the largest $w$ entries (in absolute value) replaced by $0$. The randomness used for CountSketch is independent of the randomness used to define $V_1,V_2,\ldots$.

For each $0\le i<L$, we apply CountSketch to the degrees of nodes in $V_i$ as follows.
\begin{enumerate}
\item For each edge $\{u,v\}$ inserted: If $u\in V_i$, we increment $x_u$ and if $v\in V_i$, we increment $x_v$.
\item For each edge $\{u,v\}$ deleted: If $u\in V_i$, we decrement $x_u$ and if $v\in V_i$, we decrement $x_v$.
\end{enumerate}
Note that the final value of $x_v$ is $d_v$ if $v\in V_i$ and zero otherwise. The next lemma uses properties of bounded arboricity graphs. It immediately implies that:
\begin{equation}\label{eq:tail}
 \|\mathbf{x}_{3\mu\textup{-tail}}\|_2\leq4\alpha\sqrt{|V_i|} \ .
\end{equation}
Note that a similar idea was used in Konrad et al.~\cite{00010ST24}.

\begin{lemma}\label{lem:highdegree}
At most $3\mu$ vertices have degree at least $4\alpha$.
\end{lemma}
\begin{proof}
Fix a maximum matching of size $\mu$. Let $X$ be the endpoints of this matching and let $Y$ be the set of nodes in $V\setminus X$ whose degree is at least $4\alpha$. Since $X$ is a vertex cover, all neighbors of $Y$ are in $X$, so
$4\alpha|Y|\le|E(X,Y)|$. Applying the arboricity bound to the subgraph
induced by $X\cup Y$ gives $|E(X,Y)|\le\alpha(|X|+|Y|)$. Hence
$|Y|\le2\mu/3$ and $|X|+|Y|\le3\mu$.
\end{proof}

Set $w=3K^2$ and $\delta=n^{-10}$. For $v\in V_i$, let
\[
 \widetilde d_{i,v}=\min\{n-1,\max\{0,\tilde{x}_v\}\}
\]
be the CountSketch estimate clipped to $[0,n-1]$. Since $d_v\in[0,n-1]$, clipping cannot increase the estimation error. Hence, combining \eqref{eq:cm} and \eqref{eq:tail} gives, on the assumption $\mu\le K^2$,
\begin{equation}\label{eq:degree-error}
 |\widetilde d_{i,v}-d_v|
 \le4\alpha\sqrt{|V_i|}/K
 \le\tau_i/2 \qquad \mbox{for all $v\in V_i$}
\end{equation}
with probability at least $1-\delta$, conditional on $V_i$.

We say CountSketch is \emph{successful} at level $i$ if~\eqref{eq:degree-error} holds. For the recursive estimator, we need accurate estimates of $f_v$ that are zero whenever $f_v=0$. The direct estimate $\max\{\widetilde d_{i,v}-D,0\}$ may be positive even when $f_v=0$. We therefore subtract the error term $\tau_i/2$ before taking the positive part, defining
\[
 \widetilde f_{i,v}=\max\{\widetilde d_{i,v}-D-\tau_i/2,0\}.
\]
When CountSketch is successful, $\widetilde d_{i,v}-\tau_i/2\in[d_v-\tau_i,d_v]$, and hence
\begin{equation}\label{eq:f-error}
 0\le f_v-\widetilde f_{i,v}\le\tau_i
 \qquad \mbox{for all $v\in V_i$}.
\end{equation}
In particular, $\widetilde f_{i,v}=0$ whenever $f_v=0$.

\subsection{The Recursive Estimator}\label{sec:recursive}

For the recursive estimator our focus is on estimating $F=\sum_v f_v$, where $f_v=\max(d_v-D,0)$.
To this end, we use the recursive-sketching construction of Braverman and
Ostrovsky~\cite{BravermanO13a}.

\paragraph{Recursive Estimation with Correction.}
Recall $F_L$ is computed exactly. The recursive approach will construct an estimate $\widehat F_i$ for $F_i$ in order of decreasing $i$.
Given that $\E[F_{i+1}\mid V_i]=F_i/2$, the natural recursion would estimate $F_i$ by $2\widehat F_{i+1}$.
However, even if $F_{i+1}$ were known exactly, the one-step error of this approach would be
\[
 2F_{i+1}-F_i
 =\sum_{v\in V_i}(2b_{i+1}(v)-1)f_v,
 \qquad
 \Var(2F_{i+1}-F_i\mid V_i)=\sum_{v\in V_i}f_v^2,
\]
where the variance calculation used the fact that $(2b_{i+1}(v)-1)^2=1$ and for $u\neq v$,
\[
 \E[(2b_{i+1}(u)-1)(2b_{i+1}(v)-1)]=0 \ .
\]
Note that a few large $f_v$ values can make this variance large.

To reduce this variance, we use the estimates $\widetilde f_{i,v}$ to replace the contributions of $f_v$ by the residuals $f_v-\widetilde f_{i,v}$. Define the correction term
\begin{equation}\label{eq:correction}
 C_i=\sum_{v\in V_i}a_i(v)\widetilde f_{i,v}
 \qquad \mbox{ where }\qquad a_i(v)=1-2b_{i+1}(v) \ .
\end{equation}
Starting with $\widehat F_L=F_L$, define the estimator recursively by
\[
 \widehat F_i=2\widehat F_{i+1}+C_i.
\]
The pseudo-code for the implementation is in Figure \ref{fig:trial}.

\paragraph{One-level Error.}
Assume $\mu\le K^2$ and let $E_i=2F_{i+1}+C_i-F_i$. Every $v\in V_i$ contributes $-a_i(v)f_v$ to $2F_{i+1}-F_i$ and $a_i(v)\widetilde f_{i,v}$ to $C_i$. Hence,
\begin{equation}\label{eq:actual-error-identity}
 E_i=\sum_{v\in V_i}a_i(v)(\widetilde f_{i,v}-f_v).
\end{equation}
Let $\mathcal R_i$ denote the randomness used by CountSketch at level $i$. Conditioned on $V_i,\mathcal R_i$, the coefficients in~\eqref{eq:actual-error-identity} are fixed and the signs $a_i(v)$ are fair and pairwise independent. Therefore, all cross terms vanish and
\begin{equation}\label{eq:actual-error-identity-start}
 \E[E_i^2\mid V_i,\mathcal R_i]
 =\sum_{v\in V_i}(\widetilde f_{i,v}-f_v)^2.
\end{equation}

Write $h_i=|\{v\in V_i:f_v>0\}|$. When CountSketch is successful at level $i$,~\eqref{eq:f-error} implies that at most $h_i$ terms in~\eqref{eq:actual-error-identity-start} are nonzero, and each is at most $\tau_i^2$. Therefore,
\[
 \E[E_i^2\mid V_i,\mathcal R_i]\le\tau_i^2h_i.
\]
If CountSketch is unsuccessful, clipping ensures that $f_v,\widetilde f_{i,v}\in[0,n-1]$, so~\eqref{eq:actual-error-identity-start} is at most $n^3$. Since this happens with probability at most $\delta$, conditional on $V_i$, we obtain
\begin{equation}\label{eq:actual-error-all}
 \E[E_i^2\mid V_i]\le\tau_i^2h_i+\delta n^3.
\end{equation}
\paragraph{Cumulative Error.}
Subtracting $F_i$ from the recursion gives
$\widehat F_i-F_i=2(\widehat F_{i+1}-F_{i+1})+E_i$.
Unrolling this identity, and using the fact that $\widehat F_L=F_L$, gives
\[
 \widehat F_0-F_0=\sum_{i=0}^{L-1}2^iE_i.
\]
For $L\ge1$, Cauchy--Schwarz gives
\[
 (\widehat F_0-F_0)^2
 \le L\sum_{i=0}^{L-1}4^iE_i^2.
\]
Taking expectations and using~\eqref{eq:actual-error-all} therefore gives
\[
 \E[(\widehat F_0-F_0)^2]
 \le L\sum_{i=0}^{L-1}4^i\E[\tau_i^2h_i]+\delta n^5L^2,
\]
where we used $4^i\le4^{L-1}<(n/K^2)^2\le n^2$ for $i<L$.
Since $\tau_i^2\le c_\tau^2\alpha^2(1+|V_i|/K^2)$, we obtain
\begin{align*}
 4^i\E[\tau_i^2h_i]
 &\le c_\tau^2\alpha^2p_i^{-2}
 \left(\E[h_i]+\frac{\E[|V_i| \cdot h_i]}{K^2}\right)
 =c_\tau^2\alpha^2h_0
 \left(\frac1{p_i}+\frac n{K^2}
              +\frac{1-p_i}{p_iK^2}\right)
\end{align*}
where we used the fact that pairwise-independent sampling gives $\E[h_i]=p_ih_0$ and
\begin{eqnarray*}
 \E[|V_i| \cdot h_i]
 =\sum_{u:f_u>0}\sum_{v\in V}
       \Pr[u\in V_i,\ v\in V_i]
 & =& h_0p_i\bigl((n-1)p_i+1\bigr) \\
 & =& h_0\bigl(np_i^2+p_i(1-p_i)\bigr).
\end{eqnarray*}
The term $p_i(1-p_i)$ accounts for pairs with the same vertex.
For $i<L$, we have $p_i>K^2/n$; since $K\ge1$, each of the three
terms in parentheses above is at most $n/K^2$.
Hence
\begin{equation}\label{eq:actual-total-error}
 \E[(\widehat F_0-F_0)^2]
 \le\frac{3c_\tau^2\alpha^2nh_0L^2}{K^2}+\delta n^5L^2
 \le\frac{(3c_\tau^2+1)\alpha^2nQL^2}{K^2}.
\end{equation}
Here we used $h_0\le Q$, $\delta=n^{-10}$, and $K^2<n$ when $L>0$; we assume the graph is nonempty, so $Q\ge1$. The empty graph is recognized from the edge count. If $L=0$, the estimator is exact.

\paragraph{Reducing Error Probability.} Given an estimate  $\widehat F_0
$ of $F$, we define an estimate for $Q$ as follows:
\[
 \widehat q_{\rm rec}=m-\widehat F_0 \ .
\]
By Markov's inequality and~\eqref{eq:actual-total-error},
when $\mu\le K^2$ we have \[
 \Pr\left[
   |\widehat q_{\rm rec}-Q|>
       \eta\max\{K,Q\}
 \right]
 \le
 \frac{\mathbb E[(\widehat q_{\rm rec}-Q)^2]}
      {\eta^2\max\{K,Q\}^2}                                      
 \le
 \frac{(3c_\tau^2+1) \alpha^2nL^2}{\eta^2K^3}
 \le\frac1{10} \ ,
\]
where we used  $\max\{K,Q\}^2\ge KQ$ and set the constant $C_0$ in the definition of $K$ sufficiently large.

 The error probability can be reduced using the same approach as used for the terminal estimator. Specifically, letting $\widehat Q_{\rm rec}$ be  the median of $O(\log n)$ independent copies of $\widehat q_{\rm rec}$ ensures that with high probability 
\[
 |\widehat Q_{\rm rec}-Q|
 \leq \eta\max\{K,Q\}
 \qquad\text{when }\mu\le K^2 \ ,
\] and this holds even if we let the value of $\widehat q_{\rm rec}$ be arbitrary when $|V_L|>10K^2$. Hence, $\widehat Q_{\rm rec}$ can be estimated using space $O(\log n ) \cdot (L+1) \cdot \tilde{O}(K^2) = \tilde{O}(K^2)$ 
 in the worst case.

\begin{figure}[!t]
\centering
\setlength{\fboxsep}{8pt}
\fbox{\begin{minipage}{\dimexpr\linewidth-2\fboxsep-2\fboxrule\relax}
\small
\begin{enumerate}
\item {\bf Preprocessing:}
\begin{enumerate}
\item Set $L=\max\{0,\lceil\log_2(n/K^2)\rceil\}$.
\item Let $b_1,\ldots,b_L:V\rightarrow\{0,1\}$ be pairwise-independent
hash functions, chosen independently of one another. Define $V_0=V$ and
$V_i=\{v\in V_{i-1}:b_i(v)=1\}$.
\item Compute $|V_i|$ for $0\le i\le L$. If $|V_L|>10K^2$, return
$\widehat q_{\rm term}=0$ and
$\widehat q_{\rm rec}=0$ and stop this trial.
\item For $0\le i<L$, set $a_i(v)=1-2b_{i+1}(v)$ and
$\tau_i=32\alpha\max\{1,\sqrt{|V_i|/K^2}\}$.
\end{enumerate}
\item {\bf During the stream:}
\begin{enumerate}
\item Maintain $m$ and $d_v$ for every $v\in V_L$.
\item For $0\le i<L$, use CountSketch with $w=3K^2$ and $\delta=n^{-10}$ on the degree vector masked to $V_i$, independently of $b_1,\ldots,b_L$.
\end{enumerate}
\item {\bf Postprocessing:}
\begin{enumerate}
\item Compute $F_L$ exactly from the stored degrees.
\item For each $0\le i<L$, enumerate vertices $v\in V_i$, query
CountSketch, and clip each degree estimate $\widetilde d_{i,v}$ to $[0,n-1]$.
Accumulate
\[
 C_i=\sum_{v\in V_i}a_i(v)\max\{\widetilde d_{i,v}-D-\tau_i/2,0\}
\]
during this enumeration.
\item Compute
\[
 \widehat F_0=2^LF_L+\sum_{i=0}^{L-1}2^iC_i.
\]
\item Return the basic estimates
\[
 \widehat q_{\rm term}
 =p_L^{-1}\sum_{v\in V_L}\min(d_v,D)-m,
 \qquad
 \widehat q_{\rm rec}=m-\widehat F_0.
\]
\end{enumerate}
\end{enumerate}
\end{minipage}}
\caption{Computation of $\widehat q_{\rm rec}$ and $\widehat q_{\rm term}$.}\label{fig:trial}
\end{figure}

\subsection*{AI Disclosure} ChatGPT-6 Astra was used in this research and made the connection with the recursive sketching technique. The author takes responsibility for the correctness of the paper. 

\small
\bibliographystyle{plain}
\bibliography{matching}

\end{document}